%% file: S_Maneth.tex
\documentclass[copyright,creativecommons]{eptcs}
\usepackage{breakurl}             
\usepackage{framed,amsfonts}
\usepackage{color}
\usepackage{url}
\usepackage{graphicx}
\usepackage{rotating}
\usepackage{amsmath}
\newtheorem{lemma}{Lemma}
\newtheorem{theorem}{Theorem}
\newtheorem{definition}{Definition}

\newcommand{\eop}{\hspace*{\fill}$\Box$}

\newenvironment{proof}{{\it Proof.}\quad}{\eop\vspace*{4mm}}

\newcommand{\height}[1]{{\sf height}(#1)}
\newcommand{\dom}[1]{{\sf dom}(#1)}
\newcommand{\ntrfc}{\textsc{n-t}^{\text{R}}_{\text{fc}}}
\newcommand{\trfc}{\textsc{t}^{\text{R}}_{\text{fc}}}
\newcommand{\tr}{\textsc{t}^{\text{R}}}
\newcommand{\ttt}{\textsc{t}}
\newcommand{\mtt}{\textsc{mtt}}
\newcommand{\mttlsi}{\textsc{mtt}_{\text{lsi}}}
\newcommand{\mso}{\textsc{mso}}
\newcommand{\xml}{\textsc{xml}}
\newcommand{\gsm}{\textsc{gsm}}
\newcommand{\dgsm}{\textsc{2dgsm}}
\newcommand{\io}{\textsc{io}}
\newcommand{\oi}{\textsc{oi}}
\newcommand{\dol}{\textsc{d0l}}
\newcommand{\hdtol}{\textsc{hdt{\small 0}l}}
\newcommand{\regt}{\textsc{regt}}

\newcommand{\dexptime}{\textsc{Exptime}}
\newcommand{\conexptime}{\textsc{co-NExptime}}
\newcommand{\expspace}{\textsc{Expspace}}
\newcommand{\nlogspace}{\textsc{NLogspace}}

\newcommand{\pspace}{\textsc{Pspace}}
\newcommand{\nptime}{\textsc{NPtime}}

\newcommand{\Ptime}{\textsc{Ptime}}
\newcommand{\nat}{I\!\!N}

\newcommand{\mttr}{\textsc{mtt}^{\text{R}}}

\def\rank{{\sf rank}}
\def\expand{{\sf expand}}
\def\strip{{\sf strip}}
\def\rhs{{\sf rhs}}
\def\rank{{\sf rank}}

\title{Equivalence Problems for Tree Transducers:\\A Brief Survey}
\author{Sebastian Maneth
\institute{School of Informatics\\University of Edinburgh}
\email{smaneth@inf.ed.ac.uk}
}
\def\titlerunning{Equivalence Problems for Tree Transducers: A Brief Survey}
\def\authorrunning{S. Maneth}
\begin{document}
\maketitle
\begin{abstract}
The decidability of equivalence for three important classes of tree transducers is discussed. 
Each class can be obtained as a natural restriction of deterministic macro tree transducers ($\mtt$s):
(1)~no context parameters, i.e., top-down tree transducers,
(2)~linear size increase, i.e., $\mso$ definable tree transducers, and
(3)~monadic input and output ranked alphabets.
For the full class of $\mtt$s, decidability of equivalence remains a long-standing open problem.
\end{abstract}
\section{Introduction}

The macro tree transducer ($\mtt$) was invented independently by
Engelfriet~\cite{eng80,DBLP:journals/jcss/EngelfrietV85}
and Courcelle~\cite{DBLP:journals/tcs/CourcelleF82,DBLP:journals/tcs/CourcelleF82a} 
(see also~\cite{DBLP:series/eatcs/FulopV98}).
As a model of syntax-directed translations, $\mtt$s generalize the
attribute grammars of Knuth~\cite{DBLP:journals/mst/Knuth68}. 
Note that one (annoying) issue of attribute
grammars is that they can be circular; $\mtt$s always terminate.
Macro tree transducers are a combination of context-free
tree grammars, invented by Rounds~\cite{DBLP:conf/stoc/Rounds69} and 
also known as ``macro tree grammars''~\cite{fis68}, and 
the top-down tree transducer of Rounds and 
Thatcher~\cite{DBLP:journals/mst/Rounds70,DBLP:journals/jcss/Thatcher70}:
the derivation of the grammar is (top-down) controlled by a given
input tree. In terms of a top-down transducer, the combination is
obtained by allowing nesting of state calls in the rules
(similar to the nesting of nonterminals in the productions of context-free tree grammars).
Top-down tree transducers generalize to trees the
finite state (string) transducers (also known as ``generalized sequential machines'', 
or $\gsm$s, 
see~\cite{gin66,Berstel79transductionsand}).
In terms of formal languages, compositions of $\mtt$s
give rise to
 a large hierarchy of string languages containing, e.g., the $\io$ and $\oi$ hierarchies
(at level one they include the indexed languages of Aho~\cite{DBLP:journals/jacm/Aho68}),
see~\cite{DBLP:journals/jcss/EngelfrietM02}.
$\mtt$s can be applied in many scenarios, e.g., to type check
$\xml$ transformations 
(they can simulate the $k$-pebble transducers of Milo, Suciu, and Vianu~\cite{DBLP:journals/jcss/MiloSV03}),
see~\cite{DBLP:journals/acta/EngelfrietM03,DBLP:conf/pods/ManethBPS05,DBLP:conf/icdt/ManethPS07}, or
to efficiently implement streaming XQuery transformations~\cite{icde2014,DBLP:conf/aplas/NakanoM06}.
In terms of functional programs, $\mtt$s are particularly simple
programs that do primitive recursion over an input tree and 
only produce trees as output. Applications in programming 
languages exist 
include~\cite{DBLP:journals/scp/Vogler91,DBLP:phd/de/Voigtlander2005,DBLP:conf/aplas/NakanoM06}.

Equivalence of nondeterministic transducers is undecidable,
already for restricted string transducers~\cite{DBLP:journals/jacm/Griffiths68}. 
We therefore only consider deterministic transducers.
What is known about the equivalence problem for 
deterministic macro tree transducers?
Unfortunately not much in the general case.
Only a few subcases are known to be decidable.
Here we describe three of them:
\begin{enumerate}
\item[(1)] top-down tree transducers
\item[(2)] linear size increase transducers
\item[(3)] monadic tree transducers.
\end{enumerate}
The first one was solved long ago by Esik~\cite{DBLP:journals/actaC/Esik81}, but
was revived through the ``earliest canonical normal form''
by Engelfriet, Maneth, and Seidl~\cite{DBLP:journals/jcss/EngelfrietMS09}. 
The latter implies $\Ptime$
equivalence check for total top-down tree transducers. 
The second one is solved by the decidability of equivalence 
for deterministic $\mso$ tree transducers of Engelfriet and Maneth~\cite{DBLP:journals/ipl/EngelfrietM06}.
The same authors have shown that every $\mtt$ of linear size increase
is effectively equal to an $\mso$ transducer~\cite{DBLP:journals/siamcomp/EngelfrietM03}. Hence,
decidability of equivalence follows for $\mtt$s of linear size increase.
Here we give a direct proof using $\mtt$s.
The third result is about $\mtt$s over monadic trees.
These are string transducers with copying. 
The decidability of their equivalence problem follows
through a relationship with L-systems~\cite{DBLP:journals/jcss/EngelfrietRS80};
in particular, with the sequence equivalence problem of $\hdtol$ systems.
The latter was first proved decidable by Culik~II and Karumh{\"a}ki~\cite{culkar86}.

\section{Preliminaries}

We deal with finite, ordered, ranked trees.
In such a tree, each node is labeled by a symbol from
a ranked alphabet such that the rank of the symbol is equal
to the number of children of the node.
Formally, a \emph{ranked alphabet} consists of a finite set $\Sigma$
together with a mapping $\rank_\Sigma:\Sigma\to\nat$ associating
to each symbol its rank. 
We write $\sigma^{(k)}$ to denote that the rank of $\sigma$ is equal to $k$.
By $\Sigma^{(k)}$ we denote the subset of symbols of $\Sigma$ that have
rank $k$.
Let $\Sigma$ be a ranked alphabet.
The \emph{set of all trees over $\Sigma$},
denoted $T_\Sigma$, is the smallest set of strings $T$ such that if
$k\geq 0$, $t_1,\dots, t_k\in T$, and $\sigma\in\Sigma^{(k)}$, then also
$\sigma(t_1,\dots,t_k)\in T$.
For a tree of the form $a()$ we simply write $a$.
For a set $A$, we denote by $T_\Sigma(A)$ the set of all trees over
$\Sigma\cup A$ such that the rank of each $a\in A$ is zero.
Let $a_1,\dots,a_n$ be distinct symbols in $\Sigma^{(0)}$ and 
let $t_1,\dots,t_n\in T_\Sigma$ such that none of the leaves in $t_j$ are
labeled by $a_i$ for $1\leq i\leq n$.
Then by $[a_i\leftarrow t_i\mid i\in\{1,\dots,n\}]$
we denote the \emph{tree substitution} that replaces each leaf labeled
$a_i$ by the tree $t_i$.
Thus $d(a,b,a)[a\leftarrow c(b), b\leftarrow a]$ denotes the tree
$d(c(b),a,c(b))$.

Let $t\in T_\Sigma$ for some ranked alphabet $\Sigma$.
We denote the \emph{nodes of $t$} by their Dewey dotted decimal path and
define the set $V(t)$ of nodes of $t$ as
$\{\epsilon\}\cup \{i.u\mid 1\leq i\leq k, u\in V(t_i)\}$
if $t=\sigma(t_1,\dots,t_k)$ with $k\geq 0$, $\sigma\in\Sigma^{(k)}$, and
$t_1,\dots, t_k\in T_\Sigma$. Thus, $\epsilon$ denotes the root node,
and $u.i$ denotes the $i$th child of the node $u$.
For a node $u\in V(t)$ we denote by $t[u]$ its label, and,
for a tree $t'$ we denote by $t[u\leftarrow t']$ 
the tree obtained from $t$ by replacing its subtree rooted at $u$
by the tree $t'$.
The \emph{size} of $t$, denoted $|t|$, is its number $|V(t)|$ of nodes.
The \emph{height} of $t$, denoted $\height{t}$, is defined as
$\height{\sigma(t_1,\dots,t_k)}=1+\max\{\height{t_i}\mid 1\leq i\leq k\}$ for
$k\geq 0$, 
$\sigma\in\Sigma^{(k)}$, 
and $t_1,\dots,t_k\in T_\Sigma$.

We fix the set of \emph{input variables} $X=\{x_1,x_2,\dots\}$ and the set
of \emph{formal context parameters} $Y=\{y_1,y_2,\dots\}$. 
For $n\in\nat$ we define $X_n=\{x_1,\dots,x_n\}$ and $Y_n=\{y_1,\dots, y_n\}$.

A \emph{deterministic finite-state bottom-up tree automaton} is a tuple
$A=(Q,\Sigma,\delta,Q_f)$ where $Q$ is a finite set of states, 
$Q_f\subseteq Q$ is the set of final states, and for every $\sigma\in\Sigma^{(k)}$
and $k\geq 0$, $\delta_\sigma$ is a function from $Q^k$ to $Q$.
The transition function $\delta$ is extended to a mapping from $T_\Sigma$ to $Q$ in 
the obvious way, and the set of trees accepted by $A$ is
$L(A)=\{s\in T_\Sigma\mid \delta(s)\in Q_f\}$.

\section{Macro Tree Transducers}

\begin{definition}\rm
A \emph{(deterministic) macro tree transducer} $M$ is a tuple $(Q,\Sigma,\Delta,q_0,R)$
such that $Q$ is a ranked alphabet of states with $Q^{(0)}=\emptyset$, 
$\Sigma$ and $\Delta$ are ranked alphabets of input and output
symbols, respectively, $q_0\in Q^{(1)}$ is the initial state, and
for every $q\in Q^{(m+1)}$, $m\geq 0$, and $\sigma\in\Sigma^{(k)}$, $k\geq 0$, 
the set of rules $R$ contains at most one rule of the form
\[
q(\sigma(x_1,\dots,x_k),y_1,\dots,y_m)\to t
\]
where $t\in T_{\Delta\cup Q}(X_k\cup Y_m)$ such that
if a node in $t$ has its label in $Q$, then its first child has 
its label in $X_k$. A rule as above is called a $(q,\sigma)$-rule and 
its right-hand side is denoted by $\rhs(q,\sigma)$.
\end{definition}

The translation $\tau_M: T_\Sigma\to T_\Delta$ (often just denoted $M$)
realized by an $\mtt$ $M$ is a partial function recursively defined as follows.
For each state $q$ of rank $m+1$, $M_q: T_\Sigma\to T_\Delta(Y_m)$
is the translation of $M$ starting in state $q$, i.e.,
``starting'' with a tree $q(s,y_1,\dots,y_m)$ where $s\in T_\Sigma$.
For instance, for $a\in\Sigma^{(0)}$, $M_q(a)$ simply
equals $\rhs(q,a)$.
In general, for an input tree $s=\sigma(s_1,\dots,s_k)$,
$M_q(s)$ is obtained from $\rhs(q,\sigma)$ by repeatedly replacing
a subtree of the form $q'(x_i,t_1,\dots,t_n)$ 
with $t_1,\dots t_n\in T_\Delta(Y_m)$
by
the tree $M_{q'}(s_i)[y_j\leftarrow t_j\mid 1\leq j\leq n]$.
The latter tree will also be written 
$M_{q'}(s_i,t_1,\dots,t_n)$.
We define $\tau_M=M_{q_0}$.

By the definition above, all our $\mtt$s are deterministic 
and we refer to them as ``macro tree transducers'' and
denote their class of translations by $\mtt$.
An $\mtt$ is \emph{total} if there is exactly one rule of the above form. 
If each state of an $\mtt$ is of rank one, 
then it is a \emph{top-down tree transducer}.
The class of translations realized by (deterministic) 
top-down tree transducers is denoted by $\ttt$.
In Lemmas~\ref{lm:nParikh} and~\ref{lm:ab} we 
make use of nondeterministic top-down tree transducers
and mark the corresponding class there by the letter ``N''.
A transducer is nondeterministic if there are several 
$(q,\sigma)$-rules in $R$. An $\mtt$ is \emph{monadic} if its
input and output ranked alphabets are monadic, i.e., only contain
symbols of rank $1$ and $0$.
An $\mtt$ $M$ is of \emph{linear size increase} if there is a constant $c$ 
such that $|\tau_M(s)|\leq c\cdot |s|$ for every $s\in T_\Sigma$. 

As an example, consider the transducer $M$ consisting of the rules 
in Figure~\ref{fig:mtt}.
\begin{figure}[htb]
\[
\begin{array}{lcl}
q_0(d(x_1,x_2))&\to& d(q(x_1,1(e)), q(x_2,2(e)))\\
q_0(a)&\to& a(e)\\
q(d(x_1,x_2),y_1)&\to& d(q(x_1,1(y_1)), q(x_2,2(y_1)))\\
q(a,y_1)&\to& a(y_1)
\end{array}
\]
\caption{The macro tree transducer $M$ adds reverse Dewey paths to leaves}\label{fig:mtt}
\end{figure}
The alphabets of this transducer are
$Q=\{q_0^{(1)},q^{(2)}\}$,
$\Sigma=\{d^{(2)}, a^{(0)}\}$, and 
$\Delta=\{d^{(2)}, a^{(1)}, 1^{(1)}, 2^{(1)}, e^{(0)}\}$.
The $\mtt$ $M$ translates a binary tree into the same tree, but additionally 
adds under each leaf the reverse (Dewey) path of the node.
For instance, $s=d(d(a,a),a)$ is translated into 
$d(d(a(1(1(e))),a(2(1(e)))),a(2(e)))$
as can be verified in this computation of $M$
\[
\begin{array}{l}
M(d(d(a,a),a)) = \\
d(M_q(d(a,a), 1(e)), M_q(a, 2(e))) =\\
d(d(M_q(a, 1(1(e))), M_q(a, 2(1(e)))), a(2(e))) =\\
d(d(a(1(1(e))), a(2(1(e)))), a(2(e))).
\end{array}
\]
Note that $M$ is \emph{not} of linear size increase. 
Hence, it is not MSO definable.
Since the translation is neither top-down nor monadic, it falls into
a class of $\mtt$s for which we do not know a procedure to decide
equivalence.
As an exercise, the reader may wonder whether there is an $\mtt$ that
is similar to $M$, but outputs Dewey paths below leaves (instead of their \emph{reverses}).
In contrast,
consider input trees with only exactly one $a$-leaf (and all other leaves labeled differently).
Then an $\mtt$ of linear size increase can output under the unique $a$-leaf its 
reverse Dewey path, i.e., this translation is $\mso$ definable. Is it now possible
with an $\mtt$ to output the non-reversed Dewey path?

One of the most useful properties of $\mtt$s is the effective preservation
of regular tree languages by their inverses.
This is used inside the proofs of several results presented here.
For instance, to prove that for every $\mtt$ there is an equivalent one
which is nondeleting in the parameters (= strict), one can use the above
property as follows: given a state $q$ of rank $m+1$ and a subset
$A \subseteq Y_m$, the language $T_\Delta(A)$ is regular, and hence
$M_q^{-1}(T_\Delta(A))$ is the regular set of inputs for which $q$ outputs
only parameters in the set $A$. Thus, by using regular look-ahead
we can determine which parameters are used, and can change the rules 
to call an appropriate state $q'$ which is only provided the parameters
in $A$ which it uses. Regular look-ahead is explained in Section~\ref{sect:td}.

The following result is stated in Theorem~7.4 of~\cite{DBLP:journals/jcss/EngelfrietV85}.
A similar proof as below is given at the end of~\cite{DBLP:journals/acta/EngelfrietM03}.
A slightly simpler proof, for a slightly larger class, is presented
by Perst and Seidl for macro forest transducers~\cite{DBLP:journals/ipl/PerstS04}.

\begin{lemma}\rm\label{lm:inv}
Let $M$ be an $\mtt$ with output alphabet $\Delta$ and let
$R\subseteq T_\Delta$ be a regular tree language (given by a
bottom-up tree automaton $B$). Then $\tau_M^{-1}(R)$ is effectively regular.
In particular, the domain $\dom{\tau_M}$ is effectively regular.
\end{lemma}
\begin{proof}
Let $M=(Q,\Sigma,\Delta,q_0,R)$ and $B=(P,\Delta,\delta,P_f)$.
We construct the new automaton $A=(S,\Sigma,\delta',S_f)$.
The states of $A$ are mappings $\alpha$ that associate with 
each state $q\in Q^{(m+1)}$ a mapping $\alpha(q): P^m\to P$. 
The set $S_f$ consists of all $\alpha\in S$ such that $\alpha(q_0)\in P_f$.
For $a\in\Sigma^{(0)}$ we define
$\delta'_a()=\alpha$ such that 
for every $q\in Q^{(m+1)}$ and $p_1,\dots,p_m$, 
$(\alpha(q))(p_1,\dots,p_m) = \delta^*(\rhs(q,a)[y_j\leftarrow p_j\mid j\in[m]])$.
Here, $\delta^*$ is the extension of $\delta$ to trees in $T_\Delta(P)$
by the rule $\delta(p)=p$ for all $p\in P$.
Now let $b\in\Sigma^{(k)}$ with $k\geq 1$.
For $\alpha_1,\dots,\alpha_k\in S$ we define
$\delta'_b(\alpha_1,\dots,\alpha_k)=\alpha$ where
for every $q\in Q^{(m+1)}$ and $p_1,\dots,p_m$, 
$(\alpha(q))(p_1,\dots,p_m) = \underline{\delta}^*(\rhs(q,b)[y_j\leftarrow p_j\mid j\in[m]])$.
Now, $\underline{\delta}^*$ is the extension of $\delta^*$ of above 
to symbols $q'\in Q^{(n+1)}$ by 
$\delta(q'(x_i,p_1',\dots,p_n'))=(\alpha_i(q'))(p_1',\dots,p_n')$ for
every $p_1',\dots,p_n'\in P$.
\end{proof}

\subsection{Bounded Balance}

Many algorithms for deciding equivalence of transducers are based
on the notion of bounded balance. 
Intuitively, two transducers have bounded balance if 
the difference of their outputs on any ``partial''
input is bounded by a constant.
For instance, for two $\dol$ systems $G_i=(\Sigma,h_i,\sigma)$ with $i=1,2$,
Culik~II defines~\cite{DBLP:journals/tcs/Culik76} the balance of a string $w\in\Sigma^*$ as the
difference of the lengths of $h_1(w)$ and $h_2(w)$. He shows that
equivalence is decidable for $\dol$ systems that have bounded difference.
In a subsequent article~\cite{DBLP:journals/iandc/CulikF77}, 
Culik~II and Fris show that any two equivalent $\dol$ systems
in normal form have bounded difference, thus giving the first solution to the
famous $\dol$ equivalence problem.

For a tree transducer $M$ a partial input is an input tree which contains
exactly on distinguished leaf labeled $x$, and $x$ is a fresh symbol
not in $\Sigma$. More precisely, a partial input is a tree
$p=s[u\leftarrow x]$ where $s$ is in the domain of $M$ and $u$ is a node of $s$.
Since the transducer has no rules for $x$, the computation on the
input $p$ ``blocks'' at the $x$-labeled node. Thus, the tree $M(p)$
may contain subtrees of the form $q(x,t_1,\dots,t_m)$ where $q$ is a
state of rank $m+1$. 
For instance, consider the transducer $M$ shown in the left of
Figure~\ref{fig:MN} and consider the partial input tree
$s=a(a(x))$. Then
\[
M(a(a(x))) = d(M(a(x)), M(a(x))) = d(d(q_0(x),q_0(x)), d(q_0(x),q_0(x))).
\]
It should be clear that if we replace each $q_0(x)$ by $M_{q_0}(s')$ so
that $s'\in T_\Sigma$ is a tree with $\hat{s}=s[x\leftarrow s']\in\dom{M}$,
then we obtain the output $M(\hat{s})$ of $M$ on input tree $\hat{s}$.
For instance, we may pick $s'=e$ above; then we obtain
$d(d(e,e),d(e,e))$ which indeed equals $M(a(a(e)))$.

Consider the trees $M_1(p)$ and $M_2(p)$ for two $\mtt$s $M_1$ and $M_2$.
What is the balance of these two trees?
There are two natural notions of balance: either we compare the 
sizes of $M_i(p)$, or we compare their heights.
Given two trees $t_1,t_2$ we define their \emph{size-balance} 
(for short, \emph{s-balance}) as $| |t_1| - |t_2| |$ and 
their \emph{height-balance} (for short, \emph{h-balance}) as
$| \height{s_1} - \height{s_2} |$.
Two transducers $M_1,M_2$ have \emph{bounded s-balance} (resp. h-balance)
if there exists a $c>0$ such that for any partial input $p$
the s-balance (resp. h-balance) of $M_1(p)$ and $M_2(p)$ is 
at most $c$.
Obviously, bounded h-balance implies bounded s-balance, but not vice versa.

Let $M$ and $N$ be equivalent $\mtt$s.
Do $M$ and $N$ have bounded size-balance?
To see that this is in general \emph{not} the case, it suffices to consider
two simple top-down tree transducers:
$M$ translates a monadic partial input of the form $a^n(x)$
into the full binary tree of height $n$ containing $2^n$
occurrences of the subtree $q_0(x)$. The transducer $N$ translates 
$a^n(x)$ into the full binary tree of height $n-1$ with
$2^{n-1}$ occurrences of the subtree $p(x)$.
\begin{figure}[htb]
\[
\begin{array}{llclllcl}
M:&q_0(a(x_1))&\to& d(q_0(x_1),q_0(x_1))\qquad &N:& p_0(a(x_1))&\to& p(x_1)\\
&q_0(e)&\to&e&& p_0(e)&\to&e\\
&&&&&p(a(x_1))&\to&d(p(x_1),p(x_1))\\
&&&&&p(e)&\to&d(e,e)
\end{array}
\]
\caption{Equivalent top-down tree transducers $M$ and $N$
with unbounded size-balance}\label{fig:MN}
\end{figure}
The rules of $M$ and $N$ are shown in Figure~\ref{fig:MN}.
Clearly, $M$ and $N$ are of bounded height-balance (with constant $c=1$).
But, their size-balance is not bounded.

In a similar way it can be seen that equivalent $\mtt$s need
not have bounded height-balance.
This even holds for monadic $\mtt$s:
Consider the transducers $M'$ and $N'$ with rules shown
\begin{figure}[htb]
\[
\begin{array}{llclllcl}
M':&q_0(a(x_1))&\to& q(x_1,q_0(x_1))
  &N':& p_0(a(x_1))&\to& p(x_1,p(x_1,p_0(x_1)))\\
&q_0(e)&\to&e&& p_0(e)&\to&e\\
&q(a(x_1),y_1)&\to&q(x_1,y_1)&&p(a(x_1),y_1)&\to&p(x_1,y_1)\\
&q(e,y_1)&\to&a(a(y_1))&&p(e,y_1)&\to&a(y_1)
\end{array}
\]
\caption{Equivalent monadic macro tree transducers $M'$ and $N'$
with unbounded height-balance}\label{fig:mon}
\end{figure}
in Figure~\ref{fig:mon}.
Their height-balance on input $a^n(x)$ is equal to $n$; for instance,
for the tree $a(a(x))$ we obtain
\[
\begin{array}{ll}
M'(a(a(x)))= & N'(a(a(x)))=\\
M'_q(a(x), M'(a(x)))=\qquad\qquad & N'_p(a(x), N'_p(a(x), N'(a(x))))=\\
q(x, q(x, q_0(x))) &  p(x, p(x, p(x, p(x, p_0(x)))))
\end{array}
\]
which are trees of height $3$ and $5$, respectively.
We may verify that replacing $x$ by $e$ results in equal trees:
\[
M'_q(e, M'_q(e, M(e)))=M'_q(e, M'_q(e,e))= 
M'_q(e, a(a(e)))= a(a(a(a(e)))).
\]
And for $N'$ we obtain that
\begin{multline*}
N'_p(e, N'_p(e, N'_p(e, N'_p(e, N(e)))))=
N'_p(e, N'_p(e, N'_p(e, N'_p(e,e))))=
N'_p(e, N'_p(e, N'_p(e, a(e))))=\\
N'_p(e, N'_p(e, a(a(e))))=
N'_p(e, a(a(a(e))))=
a(a(a(a(e)))).
\end{multline*}

Let us now show that equivalent top-down tree transducers 
have bounded height-balance.

\begin{lemma}\rm\label{lm:bd}
Equivalent top-down tree transducers
effectively have bounded height-balance.
\end{lemma}
\begin{proof}
Let $M_1, M_2$ be equivalent top-down tree transducers 
with sets of states $Q_1, Q_2$, respectively.
Note that they have the same domain $D$.
Let $s\in D$ and $u\in V(s)$.
Consider the two trees 
$\xi_i=M_i(s[u\leftarrow x])$ for $i=1,2$.
Let $s'$ be a smallest input tree such that 
$s[u\leftarrow s']\in D$.
It should be clear that the height of $s'$ is bounded
by some constant $d$. 
In fact, let $d$ be the height of a smallest tree in the
set $(\cap_{q\in Q}\dom{M_{1,q}}\cap(\cap_{q'\in Q'}\dom{M_{2,q'}})$,
for any subsets $Q\subseteq Q_1$ and $Q'\subseteq Q_2$. 
Since $s'$ is in such a set, its height is at most $d$.
This bound $d$ can be computed because by Lemma~\ref{lm:inv}
the sets $\dom{M_{1,q}}$ and $\dom{M_{2,q'}}$
are effectively regular, and regular
tree languages are effectively closed under intersection~\cite{tata07}.
In fact, it is not difficult to see that we can choose
$d=2^{|Q_1|+|Q_2|}$.
Hence, there is a constant $c$ such that 
$\height{M_{i,q_i}(s')}<c$ for any $q_i\in Q_i$
appearing in $\xi_i$. 
Clearly we can take $c=d\cdot h$, where $h$ is the maximal height
of the right-hand side of any rule of $M_1$ and $M_2$.
This means that $|\height{\xi_1}-\height{\xi_2}|\leq c$ because
$\xi_1\Theta_1=\xi_2\Theta_2$ and the substitutions
$\Theta_i=[q(x)\leftarrow M_{i,q}(s')\mid q\in Q_i]$
increase the height of $\xi_i$ by at most $c$.
\end{proof}

If the transducers $M_1, M_2$ of Lemma~\ref{lm:bd} are
total, then $d=1$ and $c$ is the maximal size of the right-hand
side of any rule for an input leaf symbol, i.e., 
$c=\max\{\height{\rhs(M_i,q_i,a)}\mid i\in\{1,2\}, q_i\in Q_i, a\in\Sigma^{(0)}\}$.

Let us consider an example of two equivalent 
top-down tree transducers $M$ and $N$ with
\begin{figure}[htb]
\[
\begin{array}{llclllcl}
M:&q_0(a(x_1))&\to& d(q(x_1),q_0(x_1))\qquad &N:& p_0(a(x_1))&\to& p(x_1)\\
&q_0(e)&\to&e&& p_0(e)&\to&e\\
&q(a(x_1))&\to&q'(x_1)&&p(a(x_1))&\to&d(a(p'(x_1)),p(x_1))\\
&q(e)&\to&e&&p(e)&\to&d(e,e)\\
&q'(a(x_1))&\to&a(a(q(x_1)))&&p'(a(x_1))&\to&a(p'(x_1))\\
&q'(e)&\to&a(e)&&p'(e)&\to&e
\end{array}
\]
\caption{Equivalent top-down tree transducers $M$ and $N$}\label{fig:diff}
\end{figure}
output paths of different height.
The rules of $M$ and $N$ are given in Figure~\ref{fig:diff}.
Let us consider the input tree $s=aaaa(x)$. We omit some
parentheses in monadic input trees.
We obtain
\begin{multline*}
M(s)=d(M_q(aaax),M(aaax))=
d(M_{q'}(aax),d(M_q(aax),M(aax)))=\\
d(a(a(M_q(ax))),d(M_{q'}(ax),d(M_q(ax),M(ax))))=\\
d(a(a(q'(x))),d(a(a(q(x))),d(q'(x),d(q(x),q_0(x)))))
\end{multline*}
Similarly, for the transducer $N$ we obtain
\begin{multline*}
N(s)=N_p(aaax)=d(a(N_{p'}(aax)), N_p(aax))=
d(a(a(N_{p'}(ax))), d(a(N_{p'}(ax)),N_p(ax)))=\\
d(a(a(a(p'(x)))), d(a(a(p'(x))),d(a(p'(x)),p(x)))).
\end{multline*}
As the reader may verify, if $x$ is replaced by the leaf $e$, then
indeed the output trees $M(aaaae)$ and $N(aaaae)$ are the same, i.e.,
the transducers are equivalent.
Let us compare the trees $M(aaaax)$ and $N(aaaax)$.
On the one hand, the transducer $M$ is ``ahead'' of the transducer
$N$ in the output branch $2.2.2$. It has already produced a $d$-node at
that position, while $N$ has not (and is in state $p$ at that position).
On the other hand, $N$ is ahead of $M$ at two other positions in the
output: at the node $1.1.1$ the transducer $N$ has produced an $a$-node
already, while $M$ at that node is in state $q'$, and,
at node $2.2.1$ the transducer $M$ has output an $a$-node, while
also here $M$ is in state $q'$.

\section{Decidable Equivalence Problems}\label{sect:equiv}

\subsection{Top-Down Tree Transducers}\label{sect:td}

It was shown by {\'E}sik~\cite{DBLP:journals/actaC/Esik81} 
that the bounded height-difference 
of top-down tree transducers can be used to decide equivalence. 

\begin{theorem}\rm(\cite{DBLP:journals/actaC/Esik81})\label{theo:Esik}
Equivalence of top-down tree transducers is decidable.
\end{theorem}
\begin{proof}
We follow the version of the proof given by Engelfriet~\cite{eng80}.
\begin{figure}[htb]
\centerline{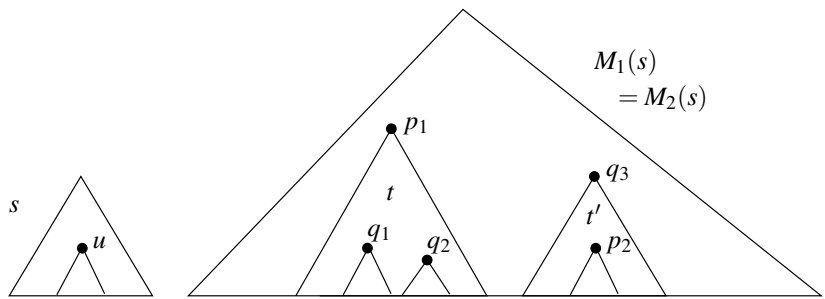}
\caption{Two equivalent top-down tree transducers}\label{fig:joost}
\end{figure}
Consider two equivalent top-down tree transducers $M_1$ and $M_2$.
By Lemma~\ref{lm:bd} they have bounded height-balance by some constant $c$.
Consider the trees $M_1(s[u\leftarrow x])$ and
$M_2(s[u\leftarrow x])$. An ``overlay'' of these two trees is
shown in Figure~\ref{fig:joost} (this is a copy of Figure~10 of~\cite{eng80}).
At the node where $M_2$ is in state $p_1$, the transducer $M_1$ has
already produced the tree $t$, i.e., at this node 
$M_1$ is ``ahead'' of $M_2$ by the amount $t$.
Similarly, at the $q_3$-labeled node, $M_2$ is ahead of $M_1$
by the amount $t'$.
Clearly, the height of $t$ and $t'$ is bounded by $c$. 
Hence, there are only finitely many such trees $t$ and $t'$.
We can construct a top-down tree automaton $A$ which in its states
keeps track of all such ``difference trees'' $t$ and $t'$, while
simulating the runs of $M_1$ and $M_2$. It checks if the outputs are
consistent, and rejects if either the outputs are different or if 
the height of a difference tree is too large.
Finally, we check if $A$ accepts the language $D=\dom{M_1}=\dom{M_2}$; this is
decidable because $D$ is regular by Lemma~\ref{lm:inv}, and equivalence
of regular tree languages is decidable (see~\cite{tata07}).
\end{proof}

Note that {\'E}sik~\cite{DBLP:journals/actaC/Esik81} shows that
even for single-valued (i.e., functional) nondeterministic top-down
tree transducers, equivalence is decidable. 
It is open whether or not equivalence is decidable for
$k$-valued nondeterministic top-down tree transducers
(but believed to be decidable along the same lines as for bottom-up
tree transducers~\cite{sei2014}, cf. the text below Theorem~\ref{theo:bu}).
A top-down tree transducer is letter-to-letter if the right-hand
side of each rule contains exactly one output symbol in $\Delta$.
It was shown by Andre and Bossut~\cite{DBLP:journals/tcs/AndreB98}
that equivalence is decidable for nondeleting \emph{nondeterministic} letter-to-letter
top-down tree transducers.
An interesting generalization of Theorem~\ref{theo:Esik} is given
by Courcelle and Franchi-Zannettacci~\cite{DBLP:journals/iandc/CourcelleF82}.
They show that equivalence is decidable for ``separated'' attribute grammars which are 
evaluated in two
independent phases: first a phase that computes all inherited
attributes, followed by a phase that computes all synthesized attributes
(top-down tree transducers are the special case of synthesized
attributes only).

\bigskip

{\bf Top-down Tree Transducers with Regular Look-Ahead.}\quad
Regular look-ahead means that the transducer ($\mtt$ or $\ttt$) comes with
a (complete) deterministic bottom-up automaton (without final
states), called the ``look-ahead
automaton of $M$''. A rule of the look-ahead transducer is of the form
\[
q(\sigma(x_1,\dots,x_k),\dots)\to t\quad\langle p_1,\dots,p_k\rangle
\]
and is applicable to an input tree $\sigma(s_1,\dots,s_k)$ only if 
the look-ahead automaton recognizes $s_i$ in state $p_i$ for all $1\leq i\leq k$.   
Given two top-down tree
transducers with regular look-ahead $M_1, M_2$, we can transform them into
ordinary transducers (without look-ahead) $N_1, N_2$ such that the resulting
transducers are equivalent if and only if the original ones are.
This is done by changing the input alphabet so that 
for every original input symbol $\sigma\in\Sigma$ of rank $k$, 
it now contains the symbols 
$\langle\sigma, p_1,\dots, p_k, q_1,\dots, q_k\rangle$
for all possible look-ahead states $p_i$ of $M_1$ and $q_i$ of $M_2$.
Thus, for every $\sigma\in\Sigma^{(k)}$, the 
new input alphabet has $|P|^k|Q|^k$-many symbols.
It is easy to see that Lemma~\ref{lm:bd} also holds 
for transducers with look-ahead, by additionally requiring
that the domain $D$ of the $N_i$ is intersected with all input trees that represent
correct runs of the look-ahead automata. 
Thus, Theorem~\ref{theo:Esik} also holds for top-down tree transducers with look-ahead. 


\begin{theorem}\rm\label{theo:dtr}
Equivalence of top-down tree transducers with
regular look-ahead is decidable.
\end{theorem}

\medskip

{\bf Canonical Normal Form.}\quad
Consider two equivalent top-down tree transducers $M_1, M_2$
and let $D$ be their domain. 
As the example transducers $M$ and $N$ with rules
in Figure~\ref{fig:diff} show, 
for a partial input tree $s[u\leftarrow x]$, 
there may be positions in the output trees
where $M_1$ is ahead of $M_2$, and other positions
where $M_2$ is ahead of $M_1$. 
Such a scenario is also depicted in Figure~\ref{fig:joost}.

We say that $M_1$ is \emph{earlier} than $M_2$,
if for every $s\in D$ and $u\in V(s)$, 
the tree $M_2(s[u\leftarrow x])$ is a prefix of the tree $M_1(s[u\leftarrow x])$.
A tree $t$ is a prefix of a tree $t'$ if for every $u\in V(t)$
with $t[u]\in\Delta$ it holds that $t'[u]=t[u]$. 
The question arises whether for every top-down translation there is an
equivalent unique earliest transducer $M$ such that
$M_q\not= M_{q'}$ for $q\not= q'$; we call such a transducer
a \emph{canonical transducer}.
The question was answered affirmative by
Engelfriet, Maneth, and Seidl~\cite{DBLP:journals/jcss/EngelfrietMS09}.
We only state this result for total transducers.

\begin{theorem}\rm\label{theo:EMS}
Let $M$ be a total top-down tree transducer. An equivalent canonical 
transducer can be constructed in polynomial time. 
\end{theorem}
\begin{proof}
The canonical transducers are top-down tree transducers 
without an initial state, but with 
an \emph{axiom tree} $A\in T_{\Delta\cup Q}(\{x_0\})$. 
This means that the translation on input tree $s\in T_\Sigma$
starts with the tree $A[x_0\leftarrow s]$ 
(instead of $q_0(s)$ for ordinary transducers).

Starting with $M$, we define its axiom $A=q_0(x_0)$.
In a first step, an earliest transducer is constructed:
if there is a state $q$ and an output symbol $\delta$ (of rank $k$)
such that $\rhs(q,\sigma)[\epsilon]=\delta$ for every input symbol $\sigma$,
then $M$ is \emph{not} earliest. 
Intuitively, the symbol $\delta$ should be produced earlier, at
each call of the state $q$. Thus, the construction 
replaces $q(x_i)$ in all right-hand sides (and in the axiom $A$) by 
$\delta(\langle q,1\rangle(x_i),\dots,\langle q,k\rangle(x_i))$
where the $\langle q,j\rangle$ are new states.
For every $\sigma$, $\rhs(\langle q,j\rangle,\sigma)$ 
is defined as the $j$-th subtree of the root of 
$\rhs(q,\sigma)$ -- beware, this right-hand side may have
changed due to the replacement above. 
Finally we remove $q$ and its rules.
This step is repeated until it cannot be applied anymore. 
In this case $M$ has become earliest and 
no $q$ and $\delta$ exists such that
$M_q(s)=\delta(\dots)$ for all input trees $s$ of $q$. 
It should be clear that the earliest step can be carried
out in polynomial time.
In the second step, equivalent states are merged to obtain the canonical transducer;
the corresponding 
equivalence relation on states is computed using
fixed point iteration in cubic time (with respect to the size of $M$).
It is computed in such a way that if $q\not= q'$, then $M_q\not= M_{q'}$.
\end{proof}

Note that the availability of a canonical (``minimal'') transducer
has many advantages. For instance, it makes possible to formulate a
Myhill-Nerode like theorem which, in turn, makes possible Gold-style
learning of top-down tree transducers (in polynomial time),
as shown by Lemay, Maneth, and Niehren~\cite{DBLP:conf/pods/LemayMN10}.

Consider the two transducers $M$ and $N$ with the rules given
in Figure~\ref{fig:diff}. To construct a canonical equivalent transducer
for $M$ according to Theorem~\ref{theo:EMS}, we observe that state
$q'$ of $M$ is \emph{not} earliest: the root equals $a$ for 
the right-hand sides of all $q$-rules. We replace 
$q'(x_1)$ by $a(\langle q',1\rangle(x_1))$ in the $(q,a)$-rule, and
introduce the two rules $\langle q',1\rangle(a(x_1))\to a(q(x_1))$ and
$\langle q',1\rangle(e)\to e$. We remove $q'$ and have obtained 
an earliest transducer. The canonical transducer is constructed
by realizing that the states $\langle q',1\rangle$ and $q$ are equivalent
and hence can be merged. The rules of the canonical transducer can
thus be given as
\[
\begin{array}{lcl}
q_0(a(x_1))&\to&d(q(x_1),q_0(x_1))\\
q_0(e)&\to&e\\
q(a(x_1))&\to&a(q(x_1))\\
q(e)&\to& e.
\end{array}
\]
To construct the canonical transducer for $N$, we observe that state $p$
is not earliest: the root equals $d$ in both rules. 
We thus replace $p(x_1)$ everywhere by $d(p_1(x_1),p_2(x_1))$ where
$p_1,p_2$ are new states. After this replacement, the 
current $(p,a)$-rule is:
\[
p(a(x_1))\to d(a(p'(x_1)),d(p_1(x_1),p_2(x_1))).
\]
Thus, new $a$-rules are
$p_1(a(x_1))\to a(p'(x_1))$ and
$p_2(a(x_1))\to d(p_1(x_1),p_2(x_1))$.
The $e$-rules are $p_1(e)\to e$ and $p_2(e)\to e$.
The resulting transducer is earliest.
We now compute that $p_1\equiv p'$ and that $p_2\equiv p_0$.
We merge these pairs of states and obtain the same transducer
(up to renaming of states) as the canonical one of $M$ above.
Hence, $M$ and $N$ are equivalent.

As a consequence of Theorem~\ref{theo:EMS}
we obtain that equivalence of total top-down
tree transducers can be decided in polynomial time. 

\begin{theorem}\rm\label{theo:complex}
Equivalence of total top-down tree transducers can be
decided in polynomial time.
\end{theorem}

The earliest normal form has also certain ``disadvantages''.
For instance, it does not preserve linearity (or nondeletingness) of the
transducer. Consider for $\Sigma=\{a^{(2)},e^{(0)}\}$ the rules
$q(a(x_1,x_2))\to d(q(x_1),q(x_2))$
and $q(e)\to d(e,e)$.
When making earliest, these rules are removed and a rule such as 
$q_1(f(x_1))\to q(x_1)$ is replaced by
$q_1(f(x_1))\to d(\langle q,1\rangle(x_1),\langle q,2\rangle(x_1))$
which is non-linear. It is also deleting:
$\langle q,1\rangle(a(x_1,x_2))\to q(x_1)$.

\subsection{Bottom-Up Tree Transducers}

As a corollary of Theorem~\ref{theo:dtr} we obtain that
also for deterministic bottom-up tree transducers,
equivalence is decidable.
This follows from the fact that every deterministic bottom-up tree transducer
can be transformed into an equivalent deterministic top-down tree transducer
with regular look-ahead~\cite{DBLP:journals/mst/Engelfriet77}.

\begin{theorem}\rm\label{theo:bu}
Equivalence of deterministic bottom-up tree transducers is decidable.
\end{theorem}
\begin{proof}
A deterministic bottom-up tree transducer is a 
tuple $B=(\Sigma,\Delta,Q,Q_f,R)$
where $Q_f\subseteq Q$ is the set of final states and $R$ contains
for every $k\geq 0$, $\sigma\in\Sigma^{(k)}$, and $q_1,\dots,q_k\in Q$
at most one rule of the form 
$\sigma(q_1(x_1),\dots,q_k(x_k))\to q(t)$ where
$t$ is a tree in $T_\Delta(X_k)$.
We construct in linear time a deterministic bottom-up 
tree automaton which for every rule
as above has the transition $\delta_\sigma(q_1,\dots,q_k)\to q$.
This automaton serves as the look-ahead automaton of a top-down
tree transducer with the unique state $p$.
For a rule as above, the transducer has the rule
\[
p(\sigma(x_1,\dots,x_k))\to t[x_i\leftarrow p(x_i)\mid i\in[k]]\quad\langle q_1,\dots,q_k\rangle.
\]
It should be clear that the resulting top-down tree transducer 
with look-ahead $T$ (which has only the single state $p$)
is equivalent to the given bottom-up tree transducer $B$.
\end{proof}

The equivalence problem for bottom-up tree transducers was first solved
by Zachar~\cite{DBLP:journals/actaC/Zachar80}.
It was shown by Seidl~\cite{DBLP:journals/tcs/Seidl92} 
that equivalence can be decided in polynomial time for
single-valued (i.e., functional) nondeterministic bottom-up tree transducers. 
Note that this also follows from Theorem~\ref{theo:j} and the (polynomial time)
construction in the proof of Theorem~\ref{theo:bu}.
This result was extended to finite-valued nondeterministic bottom-up
tree transducers by Seidl~\cite{DBLP:journals/mst/Seidl94}.
For nondeterministic letter-to-letter bottom-up tree transducers, equivalence
was shown decidable by Andre and Bossut~\cite{DBLP:conf/tapsoft/AndreB95};
such transducers contain exactly one output symbol in the right-hand side
of each rule.
They reduce the problem to the equivalence of bottom-up relabelings which
was solved by Bozapalidis~\cite{DBLP:journals/tcs/Bozapalidis92}.
For deterministic bottom-up tree transducers the effective existence of a
canonical normal form, similar in spirit to the earliest normal
form of top-down tree transducers, was shown by Friese, Seidl, 
and Maneth~\cite{DBLP:journals/ijfcs/FrieseSM11}. They show that this normal
form can be constructed in polynomial time, if each state of the given
transducer produces either none or infinitely many outputs; hence, 
equivalence is decidable in polynomial time for such transducers.
Friese presents in her PhD thesis~\cite{Friese10Thesis} a Myhill-Nerode theorem
for bottom-up tree transducers.

\subsection{Linear Size Increase mtts}

It was shown by Engelfriet and Maneth~\cite{DBLP:journals/siamcomp/EngelfrietM03}
that total deterministic mtts of linear size increase 
characterize the total deterministic
$\mso$ definable tree translations. In fact, even any composition of total deterministic
mtts, when restricted to linear size increase, is equal to an $\mso$ definable translation,
as shown by Maneth~\cite{DBLP:conf/fsttcs/Maneth03}.
The $\mso$ definable tree translations are a special instance of the
$\mso$ definable graph translations, introduced by Courcelle and
Engelfriet, see~\cite{DBLP:books/daglib/0030804}.
Decidability of equivalence for deterministic $\mso$ graph-to-string translations 
on a context-free graph language was proved by Engelfriet and Maneth~\cite{DBLP:journals/ipl/EngelfrietM06}.
It implies decidable equivalence also for $\mso$ tree translations.
We present a proof of the latter here that only uses $\mtt$s and avoids going through $\mso$. 

The idea of the proof stems from Gurari's proof~\cite{DBLP:journals/siamcomp/Gurari82} of
the decidability of equivalence for $\dgsm$.
In a nutshell: the ranges of all the above translations are Parikh.
A language is \emph{Parikh} if its set of Parikh vectors is equal to the
set of Parikh vectors of a regular language. 
Let $\Sigma=\{a_1,\dots,a_m\}$ be an alphabet.
The \emph{Parikh vector} of a string $w\in\Sigma^*$ 
is the $n$-tuple $(i_1,\dots,i_m)$ of natural numbers $i_j$ such that
for $1\leq j\leq m$, $i_j$ equals the number of occurrences of $a_j$ in $w$.
For a language that is Parikh, it is decidable whether or not it contains
a string with Parikh vector $(n,n,\dots,n)$ for some natural number $n$.
This property is used to prove equivalence as follows.
Given two tree-to-string transducers $M_1, M_2$ we first change $M_i$ to 
produce a new end marker $\$$ at the end of each output string. 
Then, given the regular domain language $D$ of $M_1$ and $M_2$,
and two distinct output letters $a,b$
we construct the Parikh language
\[
L^{a,b}=\{a^mb^n\mid \exists s\in D: M_1(s)/m=a, M_2(s)/n=b\}.
\]
Here $w/m$ denotes the $m$-th letter in the string $w$.
We now decide if there is an $n$ such that $a^nb^n\in L^{a,b}$, using the fact that $L^{a,b}$ is Parikh.
If such an $n$ exists, then the transducers $M_1, M_2$ are \emph{not} equivalent.
If, for all possible $a$, $b$, no such $n$ exists, then we know
that the transducers $M_1, M_2$ are equivalent.

It was shown by Engelfriet, Rozenberg, and Slutzki
in Corollary~3.2.7 of~\cite{DBLP:journals/jcss/EngelfrietRS80}
that ranges of \emph{nondeterministic} finite-copying top-down
tree transducers with regular look-ahead (for short, $\ntrfc$s) possess the Parikh property.
The nondeterminism of this result is useful for defining
the language $L^{a,b}$, because we need to nondeterministically choose $a$ and
$b$ positions $m$ and $n$ of the output strings. 
A nondeterministic top-down tree transducer $M$ is \emph{finite-copying}
if there is number $c$ such that for every $s\in T_\Sigma$ and $u\in V(s)$,
the number of occurrences of states (more precisely, subtrees $q(x)$ such
that $q$ is a state of $M$)
in the tree $M(s[u\leftarrow x])$ is $\leq c$.
We denote the class of translations of nondeterministic finite-copying
top-down tree transducers by $\ntrfc$.

For a tree $t$ we denote by $yt$ its \emph{yield}, i.e., 
the string of its leaf labels from left to right.
For a class $X$ of tree translations we denote by
$yX$ the corresponding class of \emph{tree-to-yield} translations.
The tree-to-yield translations of top-down tree transducers can 
be obtained by top-down tree-to-\emph{string} transducers
which have strings over output symbols and state calls $q(x_i)$ 
in the right-hand sides of their rules.
We repeat the argument given in~\cite{DBLP:journals/jcss/EngelfrietRS80}.
By $\regt$ we denote the class of regular tree languages, i.e.,
those languages recognized by
(deterministic) finite-state bottom-up tree automata.

\begin{lemma}\rm\label{lm:nParikh}
Languages in $y\ntrfc(\regt)$ are Parikh.
\end{lemma}
\begin{proof}
Let $M$ be a $y\ntrfc$ transducer and let $R\in\regt$.
A top-down transducer is \emph{linear} if no $x_i$ appears more than once
in any of the right-hand sides of its rules.
We construct a linear
transducer $M'$ such that $\dom{M'}=\dom{M}$ and 
the string $M'(s)$ is a permutation of the string $M(s)$, for every $s\in\dom{M}$.
The new transducer computes in its states the state sequences of $M$, i.e.,
the sequence of states that are translating the current input node.
Since $M$ is finite-copying, there effectively exists a bound $c$ on the
length of the state sequences. For a new state $\langle q_1,\dots, q_n\rangle$
with $n\leq c$ the right-hand side of a rule is obtained by simply concatenating the 
right-hand sides of the corresponding rules for $q_i$.
It is well known that linear top-down tree transducers preserve regularity 
and hence the language $M'(R)$ is in $y\regt$, i.e., 
it is the yield language of a regular tree language.
The latter is obviously a context-free language 
(cf. Theorem~3.8 of~\cite{DBLP:journals/jcss/Thatcher70}) which is
Parikh by Parikh's theorem~\cite{DBLP:journals/jacm/Parikh66}.
\end{proof}

A macro tree transducer $M$ is \emph{finite-copying} if there exist constants $k$ and $n$ 
such that 
\begin{enumerate}
\item[(1)] for every input tree $s'=s[u\leftarrow x]$ with $s\in T_\Sigma$ and $u\in V(s)$, 
the number of occurrences of states in $M(s')$ is $\leq k$ and
\item[(2)] for every state $q$ of rank $m+1$, $1\leq j\leq m$, and $s\in T_\Sigma$, the number of occurrences
of $y_j$ in $M_q(s)$ is $\leq n$.
\end{enumerate}

Recall that an $\mtt$ $M$ is of linear size increase if there is a constant $c$ 
such that $|\tau_M(s)|\leq c\cdot |s|$ for every $s\in T_\Sigma$. 
We denote the class of translations realized by $\mtt$s of linear size increase
by $\mttlsi$.

\begin{lemma}\rm\label{lm:ymtt}
$(y\mttlsi)\subseteq y\trfc$.
\end{lemma}
\begin{proof}
It was shown in~\cite{DBLP:journals/siamcomp/EngelfrietM03} 
how to construct a finite-copying macro tree transducer
with look-ahead, for a given macro tree transducer of linear size increase.
The construction goes through several normal forms which make sure that
the transducer generates only finitely many copies; most essentially, the ``proper''
normal form: each state produces infinitely many output trees, and, each parameter
is instantiated by infinitely many trees. By using regular look-ahead finitely
many different trees can be determined and outputted directly.
The idea of the proper normal form was used already by 
Aho and Ullmann for top-down tree transducers~\cite{DBLP:journals/iandc/AhoU71}.

It was shown in Lemmas~6.3 and~6.6 of~\cite{DBLP:journals/iandc/EngelfrietM99}
that $M$ can be changed into an equivalent transducer which is
``special in the parameters''. This means that it is linear and nondeleting in
the parameters, i.e., each parameter $y_j$ of a state $q$ appears \emph{exactly once} 
in the right-hand side of each $(q,\sigma)$-rule.
The idea is to simply provide multiple parameters, whenever parameters are copied,
and to use regular look-ahead in order to determine which parameters are deleted. 
This was mentioned above Lemma~\ref{lm:inv} already. 
For a $y\mttr$ transducer that is special in the parameters, it was shown in
Lemma~13 of~\cite{DBLP:journals/jcss/EngelfrietM02} how to construct an equivalent
$y\tr$ transducer. The parameters of the $y\mtt$ can be removed by outputting the
strings between them directly.  Since each parameter appears once, the final string 
$M_q(s)$ is divided into $m+1$ string chunks $w_j$ (where $m+1$ is the rank of $q$):
$w_0,y_1w_1,\dots,y_mw_m$. We leave further details as an exercise,
and suggest to start with the case that all $y_j$ appear in strictly increasing
order at the leaves of any $M_q(s)$.
It is not difficult to see that the construction preserves 
finite-copying.
\end{proof}

\begin{lemma}\rm\label{lm:ab}
Let $M_1, M_2$ be $y\trfc$ transducers with input and output alphabets $\Sigma$ and $\Delta$, 
and let $a,b\in\Delta$ with $a\not=b$.
Let $D\subseteq T_\Sigma$ be a regular tree language.
The language $L^{a,b}=\{a^m\# b^n\mid \exists s\in D: M_1(s)/m=a, M_2(s)/n=b\}$ is Parikh.
\end{lemma}
\begin{proof}
Let us assume that the state sets $Q_1, Q_2$ of the transducers $M_1,M_2$ are
disjoint. The initial state of $M_1, M_2$ is $q_0$ and $p_0$, respectively.
We first construct a $y\ntrfc$ transducer $M_1'$ such that 
\[
M_1'(s)=\{ua\mid u\in\Delta^*, \exists v\in\Delta^*: uav = M_1(s)\}.
\]
Its state set is $Q_0=Q_1\cup \{q_a\mid q\in Q_1\}$ and its initial state is $q_{0,a}$.
It has all rules of $M_1$ and, moreover, for every rule 
$q(\sigma(x_1,\dots,x_k))\to w\quad\langle\cdots\rangle$
of $M_1$, whenever $w=uav$ it has the rule 
$q_a(\sigma(x_1,\dots,x_k))\to ua\quad\langle\cdots\rangle$, and whenever
$w=uq'(x_i)v$ it has the rule  
$q_a(\sigma(x_1,\dots,x_k))\to uq_a'(x_i)\quad\langle\cdots\rangle$.
From $M_1'$ one obtains a $y\ntrfc$ transducer $M_1''$ such that
\[
M_1''(s)=\{a^m\mid M_1(s)/m=a\}
\]
by simply changing all symbols of $\Delta$
into $a$  in the rules of $M_1'$. Similarly, one obtains a transducer $M_2''$ 
such that $M_2''(s)=\{b^n\mid M_2(s)/n=b\}$.
Finally, a $y\ntrfc$ transducer $M$ is defined such that 
$M(s)=\{ a^m\# b^n\mid M_1(s)/m=a, M_2(s)/n=b\}$.
Its state set is $\{ r_0\}\cup Q_1' \cup Q_2'$ with initial state $r_0$. 
The look-ahead automaton of $M$ is the product automaton of the
look-ahead automata of $M_1$ and $M_2$. 
The set of rules of $M$ is the union of those of $M_1''$ and $M_2''$, adapted 
to the new look-ahead appropriately. 
Moreover, for $\sigma\in\Sigma^{(k)}$, $k\geq 0$, and rules
$q_{0,a}(\sigma(x_1,\dots,x_k))\to u\quad\langle q_1',\dots,q_k'\rangle$ and
$p_{0,b}(\sigma(x_1,\dots,x_k))\to w\quad\langle p_1',\dots,p_k'\rangle$, we let 
\[
r_0(\sigma(x_1,\dots,x_k))\to u\# w\quad\langle (q_1',p_1'),\dots,(q_k',p_k')\rangle
\]
be a rule of $M$. 
Obviously, $M(s)$ equals the concatenation $M_1''(s)\# M_2''(s)$,
and is finite-copying.
Since $M(D)=L^{a,b}$ it follows by Lemma~\ref{lm:nParikh} that $L^{a,b}$ is Parikh.
\end{proof}

\begin{theorem}\label{theo:mttlsi}
Equivalence of deterministic macro tree transducers of linear size increase
is decidable.
\end{theorem}
\begin{proof}
Let $M_1, M_2$ be $\mtt$ transducers of linear size increase.
We first check that the domains of $M_i$ coincide. This is decidable
because $\dom{M_i}$ is effectively regular by Lemma~\ref{lm:inv}.
If not then the transducers are not equivalent and we are finished.
Otherwise, let $D$ be their domain. 
We may consider $M_i$ as tree-to-string transducers, by considering
the tree in the right-hand side of each rule as a string (which
uses additional terminals symbols for denoting the tree structure
such as opening and closing parentheses and commas). 
Thus, by Lemma~\ref{lm:ymtt} (which is effective) we may in fact assume 
that $M_1$ and $M_2$ are $y\trfc$ transducers.
Let $\Delta$ be
the output alphabet of $M_i$ and let $\$$ be a new symbol not in $\Delta$. 
We change $M_i$ so that each output string is followed by the $\$$ symbol.
This can easily be done by first splitting the initial state $q_0$ so that it appears
in the right-hand side of no rule, and then adding $\$$ to the end
of each $q_0$-rule. It now holds that $M_1$ and $M_2$ are \emph{not equivalent} if and only
if there exist $a,b\in\Delta$ with $a\not=b$, $s\in D$, and a number $n$ such that
$M_1(s)/n=a$ and $M_2(s)/n=b$. The latter holds if the intersection of $L^{a,b}$ of Lemma~\ref{lm:ab}
with the language $E=\{a^n\# b^n\mid n\in\nat\}$ is nonempty.
Since $L^{a,b}$ is Parikh by Lemma~\ref{lm:ab}, we obtain decidability because
semilinear sets are closed under intersection~\cite{ginspa64,gin66} and have decidable emptiness.
But, there is a much easier proof: $E\cap L$ is context-free, because $E$ is
(by the well-known ``triple construction'', see, e.g.,
Theorem~6.5 of~\cite{DBLP:books/aw/HopcroftU79}),
where $L$ is a regular language with the same Parikh vectors as $L^{a,b}$.
The result follows since context-free grammars have decidable emptiness.
\end{proof}

\subsection{Monadic mtts}

Recall that a macro tree transducer is monadic if both its input 
and output alphabet are monadic, i.e., consist of symbols of
rank one and rank zero only.
We will reduce the equivalence problem for monadic $\mtt$ transducers
to the sequence equivalence problem of $\hdtol$ systems.
An $\mtt$ is \emph{nondeleting} if for every state $q$ of rank $m+1$,
$1\leq j\leq m$, and input symbol $\sigma$, the parameter $y_j$ occurs in $\rhs(q,\sigma)$.
A monadic $\mtt$ $M=(Q,\Sigma,\Delta,q_0,R)$ is \emph{normalized} if
\begin{enumerate}
\item[(N0)] it is nondeleting 
\item[(N1)] each state is of rank two or one, i.e., $Q=Q^{(2)}\cup Q^{(1)}$ and
\item[(N2)] there is only one input and output symbol of rank zero, i.e.,
$\Sigma^{(0)}=\Delta^{(0)}=\{\bot\}$
\end{enumerate}
Note that for total transducers (N1) is a consequence of (N0) because
a $(q,\bot)$-rule can only contain at most one parameter occurrence.

\medskip

{\bf $\hdtol$ systems}\quad
An instance of the $\hdtol$ sequence equivalence problem consists of finite
alphabets $\Sigma$ and $\Delta$, two strings $w_1,w_2\in\Sigma^*$,
homomorphisms $h_j,g_j:\Sigma^*\to\Sigma^*$, $1\leq j\leq n$, and
homomorphisms $h,g:\Sigma^*\to\Delta^*$. To solve the problem we have to determine
whether or not
\[
h(h_{i_k}(\cdots h_{i_1}(w_1)\cdots))=
g(g_{i_k}(\cdots g_{i_1}(w_2)\cdots))
\]
holds true for all $k\geq 0$, $1\leq i_1,\dots,i_k\leq n$.
This problem is known to be decidable. It was first proven 
by Culik~II and Karhum{\"a}ki~\cite{culkar86}, using Ehrenfeucht's Conjecture
and Makanin's algorithm. A later proof of Ruohonen~\cite{ruh86} is based
on the theory of metabelian groups. Yet another, very short, proof
was given by Honkala~\cite{DBLP:journals/tcs/Honkala00a} which only relies on Hilbert's Basis Theorem.
We now show that the equivalence problem for total monadic $\mtt$
transducers can be reduced to the sequence equivalence problem for $\hdtol$ systems. 
For a monadic tree $s=a_1(\cdots a_n(e)\cdots)$ we denote
by $\strip(s)$ the string $a_1\cdots a_n$. 

\begin{lemma}\label{lm:totalmonadic}
Equivalence of total monadic normalized $\mtt$s 
on a regular input language is decidable.
\end{lemma}
\begin{proof}
We first solve the problem without a given input tree language. 
Let $M_1=(Q_1,\Gamma,\Pi,q_0,R_1)$ and $M_2=(Q_2,\Gamma,\Pi,p_0,R_2)$
be total monadic normalized macro tree transducers such 
that $Q_1$ is disjoint from $Q_2$. Let $Q=Q_1\cup Q_2$. 
We define an instance of the $\hdtol$ sequence equivalence problem.
The string alphabets $\Sigma, \Delta$ are defined as $\Sigma=\Pi^{(1)}\cup Q$
and $\Delta=\Pi^{(1)}$.
We define homomorphisms $h_a, g_a$ for every input symbol $a\in\Gamma^{(1)}$.
For $\pi\in\Pi^{(1)}$ let $h_a(\pi)=g_a(\pi)=\pi$.
Let $q\in Q$. 
If $q\in Q_1$ then let $h_a(q)=\strip(\rhs_{M_1}(q,a))$,
and otherwise let $h_a(q)=q$. 
If $q\in Q_2$ then let $g_a(q)=\strip(\rhs_{M_2}(q,a))$,
and otherwise let $g_a(q)=q$. 
For trees $t\in T_{\Pi\cup Q}(X_k\cup Y_m)$ we define the mapping $\strip$ 
by $\strip(\pi(t))=\pi\cdot\strip(t)$ for $\pi\in\Gamma^{(1)}$,
$\strip(q(x_1,t))=q\cdot\strip(t)$ for $q\in Q^{(2)}$, 
$\strip(q(x_1))=q$ for $q\in Q^{(1)}$, and 
$\strip(\bot)=\strip(y_1)=\epsilon$, where ``$\cdot$''
denotes string concatenation.
The final homomorphisms $h, g$ are defined as $h(q)=\strip(\rhs_{M_1}(q,\bot))$
if $q\in Q_1$, and otherwise $h(q)=q$, and
$g(q)=\strip(\rhs_{M_2}(q,\bot))$ if $q\in Q_2$, and otherwise $g(q)=q$.
Last but not least, let $w_1=q_0$ and $w_2=p_0$. This ends the construction of the $\hdtol$ instance.
Consider an input tree $s=a_1(\cdots a_n(\bot)\cdots)\in T_\Gamma$.
It should be clear that 
$h(h_{a_n}(\cdots h_{a_1}(w_1)\cdots))=\strip(M_1(s))$ and that
$g(g_{a_n}(\cdots g_{a_1}(w_2)\cdots))=\strip(M_2(s))$.
Thus, this instance of the $\hdtol$ sequence equivalence problem 
solves the equivalence problem of the two transducers $M_1$ and $M_2$.

Let $D\subseteq T_\Gamma$ be a regular input tree language.
We wish to decide whether $M_1(s)=M_2(s)$ for every $s\in D$.
We assume that $D$ is given by a deterministic finite-state automaton $A$ that runs top-down on the
unary symbols in $\Gamma^{(1)}$. We further assume that 
$A=(R,\Gamma^{(1)},r_0,\delta,R_f)$ is complete, i.e., 
for every state $r\in R$ and every symbol $a\in\Gamma^{(1)}$,
$\delta(r,a)$ is defined (and in $R$). Note that $r_0$ is the initial state and $R_f\subseteq R$ is
the set of final states. 
Let $\Sigma=\Pi^{(1)}\cup Q$ as before and define
$\Sigma'=\{\langle r,b\rangle\mid r\in R, b\in\Sigma\}$ and $\Delta=\Pi^{(1)}$.
Our $\hdtol$ instance is over $\Sigma'$ and $\Delta$.
Let $a\in\Gamma^{(1)}$, $r\in R$, and $r'=\delta(r,a)$.
For $\pi\in\Pi^{(1)}$ let 
$h_a(\langle r,\pi\rangle)=g_a(\langle r,\pi\rangle)=\langle r',\pi\rangle$.
Let $q\in Q_1$ and $p\in Q_2$.
Define
\[
\begin{array}{lcl}
h_a(\langle r,q\rangle )&=&\strip(\rhs_{M_1}(q,a))[b\leftarrow\langle r',b\rangle\mid b\in\Sigma]\\
g_a(\langle r,p\rangle )&=&\strip(\rhs_{M_2}(p,a))[b\leftarrow\langle r',b\rangle\mid b\in\Sigma].
\end{array}
\]
Let $h_a(\langle r,p\rangle)=\langle r,p\rangle$ and 
$g_a(\langle r,q\rangle)=\langle r,q\rangle$.
The final homomorphisms $g, h$ are defined as follows.
If $r\in R_f$ then let
$h(\langle r,q\rangle)=\strip(\rhs_{M_1}(q,\bot))$,
$g(\langle r,p\rangle)=\strip(\rhs_{M_2}(p,\bot))$, and 
let $h(\langle r,b\rangle)=b$ and  $g(\langle r,b\rangle)=b$
for the remaining cases.
If $r\not\in R_f$ then let
$h(\langle r,b\rangle) = g(\langle r,b\rangle) = \epsilon$
for every $b\in\Sigma$.
The initial strings are defined as $w_1=\langle r_0,q_0\rangle$ and
$w_2=\langle r_0,p_0\rangle$.

Consider an input tree $s=a_1(\cdots a_n(\bot)\cdots)\in T_\Gamma$
and let $1\leq j\leq n$.
It should be clear that if $\delta^*(r_0,a_1\cdots a_j)=r$, i.e.,
$A$ arrives in state $r$ after reading the prefix $a_1\cdots a_j$,
then 
\[
h_{a_j}(\cdots h_{a_1}(w_1)\cdots ) = 
\strip(M_1(a_1\cdots a_j(x)))[\pi\leftarrow\langle r,\pi\rangle\mid\pi\in\Delta]
[q\leftarrow\langle r,q\rangle\mid q\in Q_1]
\]
and similarly for $g$ and $M_2$. Thus each and every symbol of a sentential form
is labeled by the current state of the automaton $A$.
Hence, if $s\not\in D$, then every symbol in 
$u_1=h_{a_n}(\cdots h_{a_1}(w_1)\cdots)$ and in 
$u_2=g_{a_n}(\cdots g_{a_1}(w_2)\cdots)$ is labeled by some state $r\not\in R_f$.
This implies that $h(u_1) = g(u_2) =\epsilon$, i.e., the final strings are
equal whenever $s\not\in D$. If on the contrary $s\in D$ then 
every symbol in $u_i$ is labeled by a final state and therefore
$h(u_i)=\strip(M_i(s))$ as before.
\end{proof}

Input and output symbols of rank zero of a given transducer
become symbols of rank one in the corresponding normalized transducer.
For a monadic tree $t=a_1(\cdots a_n(e)\cdots)$ we denote by
$\expand(t)$ the tree $a_1(\cdots a_n(e(\bot))\cdots)$.

\begin{lemma}\label{lm:expand}
For every monadic $\mttr$ transducer $M$ a normalized $\mttr$ transducer $N$ can
be constructed such that $\tau_N=\{(\expand(s),\expand(t))\mid (s,t)\in\tau_M\}$.
\end{lemma}
\begin{proof}
Using regular look-ahead we first make $M$ nondeleting.
As mentioned in the proof of Lemma~\ref{lm:ymtt}, this construction was
given in the proof of Lemma~6.6 of~\cite{DBLP:journals/iandc/EngelfrietM99}.
Now, every parameter that appears in the left-hand side of a rule, also
appears in the right-hand side. Since the final output tree is monadic,
the resulting transducer satisfies (N1) above. 
Finally, we define the $\mttr$ transducer $N$ which has input and output alphabets
$\Sigma'=\Sigma^{(1)}\cup\{a'^{(1)}\mid a\in\Sigma^{(0)}\}$ and
$\Delta'=\Delta^{(1)}\cup\{a'^{(1)}\mid a\in\Delta^{(0)}\}$.
For input symbols in $\Sigma^{(1)}$ the transducer $N$ has exactly
the same rules as $M$. 
Let $q\in Q$ and $a\in\Sigma^{(0)}$ such that $\rhs(q,a)$ is defined.
Then we let 
\[
q(a'(x_1))\to\rhs(q,a)[b\leftarrow b'(\bot)\mid b\in\Delta^{(0)}].
\]
be a rule of $N$. 
Regular look-ahead can be used to ensure that only trees of
the form $\expand(s)$ are in the domain of $N$.
\end{proof}

Obviously, two monadic $\mtt$ transducers are equivalent if and only 
if their normalized versions are equivalent. Hence, it suffices to 
consider the equivalence problem of normalized monadic $\mtt$ transducers.

\begin{theorem}
Equivalence of monadic macro tree transducers with regular look-ahead is decidable.
\end{theorem}
\begin{proof}
Let $M_1,M_2$ be monadic macro tree transducers with regular look-ahead
and let $\Sigma$ be their input alphabet.
Let $A_1, A_2$ be the look-ahead automata of $M_1, M_2$. 
By Lemma~\ref{lm:expand} we may assume that $M_1$ and $M_2$ are normalized.
We first check if the domains of $M_1$ and $M_2$ coincide. 
If not then the transducers are not equivalent and we are finished.
Otherwise, let $D$ be their domain. 
We define two total monadic $\mtt$s $N_1, N_2$ \emph{without} look-ahead. 
Let $P_1,P_2$ be the sets of states of $A_1,A_2$, respectively.
The input alphabet of $N_i$ is defined as
$\Sigma'=\{\langle\sigma,p_1,p_2\rangle\mid\sigma\in\Sigma^{(1)},p_1\in P_1, p_2\in P_2\}$. 
An input symbol $\langle\sigma, p_1,p_2\rangle$ denotes that the look-ahead automata
at the child of the current node are in states $p_1$ and $p_2$, respectively.
Thus, the $(q,\langle\sigma,p_1,p_2\rangle)$-rule of $N_1$ is defined as 
the $(q,\sigma)$-rule with look-ahead $\langle p_1\rangle$ of $M_1$,
and the $(q,\langle\sigma,p_1,p_2\rangle)$-rule of $N_2$ is defined as 
the $(q,\sigma)$-rule with look-ahead $\langle p_2\rangle$ of $M_2$.
Finally, we make $N_1$ and $N_2$ total (in some arbitrary way).

For a tree $t$ in $T_{\Sigma'}$ we denote by $\gamma(t)$ the tree in $T_\Sigma$ obtained
by changing every label $\langle \sigma,p,p'\rangle$ into the label $\sigma$. 
Let $E\subseteq T_{\Sigma'}$ be the regular tree language consisting of all trees
$t$ such that 
\begin{enumerate}
\item[(1)] $s=\gamma(t)$ is in $D$, 
\item[(2)] the second components of the labels in $t$ constitute a correct run of $A_1$ on $s$, and
\item[(3)] the  third components of the labels in $t$ constitute a correct run of $A_2$ on $s$.
\end{enumerate}
Clearly, for the resulting transducers $N_i$ it holds 
that $N_1$ and $N_2$ are equivalent on $E$ if and only if 
$M_1$ is equivalent to $M_2$. 
Hence decidability of equivalence follows from Lemma~\ref{lm:totalmonadic}.
\end{proof}

Note that macro tree transducers with monadic output alphabet are 
essentially the same as top-down tree-to-string transducers (see Lemma~7.6
of~\cite{DBLP:journals/iandc/EngelfrietM99}).
For the latter, the equivalence problem was stated already in 1980
by Engelfriet~\cite{eng80} as a big open problem.
This problem remains open, but, as this section has shown, at least
for the restricted case of monadic input, we obtain decidability.
Note further that the connection between L-systems and tree transducers
is well known and was studied extensively in~\cite{DBLP:journals/jcss/EngelfrietRS80}.

\section{Complexity}

In Section~\ref{sect:equiv} we already mentioned one
complexity result, viz. Theorem~\ref{theo:complex}, which states 
that equivalence can be decided in
polynomial time for total  top-down tree transducers.
How about top-down tree transducers ($\ttt$s) in general?
It was mentioned in the Conclusions of~\cite{DBLP:conf/mfcs/BenediktEM13}
that checking equivalence of $\ttt$s can be done in double exponential time,
using the procedure of~\cite{DBLP:journals/jcss/EngelfrietMS09}.

Without giving details we now present a proof that strengthens
both results above (and which also works for transducers with look-ahead).
We show that equivalence for $\ttt$s can be decided in $\expspace$,
and for total $\ttt$s in $\nlogspace$.
For a top-down tree transducer $M$ and trees $s,t$ with $t=M(s)$,
it holds that each node $v$ in the output tree $t$ is produced by one
particular node $u$ in $s$. The latter is called $v$'s origin. 
It means that $M(s[u\leftarrow x])$ does not have a $\Delta$-node $v$,
while $v$ is a $\Delta$ node in $M[s\leftarrow a(x,\dots,x)]$ where $a=s[u]$.

\begin{theorem}\rm\label{theo:j}
Equivalence of top-down tree transducers with
regular look-ahead is decidable in $\expspace$, and for
total transducers in $\nlogspace$.
\end{theorem}
\begin{proof}
We sketch the proof for transducers without look-ahead.
Since both complexity classes are closed under complement,
it suffices to consider nonequivalence.
Consider two top-down tree transducers $M_1$ and $M_2$. 
The idea (as in the finite-copying case)
is to guess (part of) an input tree $s$ and a node $v$ of the output trees
$t_1=M_1(s)$ and $t_2=M_2(s)$ such that $t_1[v]\neq t_2[v]$. 
It suffices to guess
the two \emph{origins} of $v$ with respect to $M_1$ and $M_2$: nodes $u_1$ and $u_2$ 
of $s$, respectively.
More precisely, it suffices to guess the paths from the root of $s$
to $u_1$ and $u_2$, and the path from the root of $t_1$ and $t_2$ to $v$,
where we may assume that all proper ancestors of $v$ have the same label
in $t_1$ and $t_2$. When guessing the path from the root of $s$ to the least
common ancestor of $u_1$ and $u_2$, the path in $t_1$ can be ahead of the path
in $t_2$, or vice versa, so the difference between these paths must be
stored. But it suffices to keep the length of this difference to be
at most exponential in the sizes of $M_1$ and $M_2$, due to the bounded
height-balance of $M_1$ and $M_2$ in case they are equivalent. 
In the proof of Lemma~\ref{lm:bd}
the height of the smallest tree $s'$ is at most exponential, and hence
the height of its translation is at most exponential. Hence the
difference between the paths in $t_1$ and $t_2$ can be stored in exponential
space. 

If $M_1$ and $M_2$ are total, then the difference between the paths in $t_1$
and $t_2$ is at most a path in a right-hand side of a rule, which can be kept in
logarithmic space. Logarithmic space is also needed to do all the
guesses, of course.
The same proof as above also holds for transducers with regular look-ahead.
\end{proof}

\begin{theorem}\label{theo:hard}
Equivalence of top-down tree transducers is $\dexptime$-hard.
\end{theorem}
\begin{proof}
It is well known that testing intersection emptiness of $n$ deterministic
top-down tree automata $A_1,\dots,A_n$ is $\dexptime$-complete. This was shown 
by Seidl~\cite{DBLP:journals/ipl/Seidl94}, cf. also~\cite{tata07}.
Let $\Sigma$ be the ranked alphabet of the $A_i$. 
We define the top-down tree transducer $M_1=(\{q_0,\dots, q_n\},\Sigma,\Sigma\cup\{\delta^{(n)}\},q_0,R)$.
We consider each $A_i$ as a partial identity transducer with start state $q_i$,
and add the corresponding rules to $R$. Thus, $M_{1,q_i}=\{(s,s)\mid s\in L(A_i)\}$.
Let $\sigma\in\Sigma$ be an arbitrary symbol of rank $\geq 1$,
and let $e$ be an arbitrary symbol in $\Sigma^{(0)}$.
We add these two rules to $R$:
\[
\begin{array}{lcl}
q_0(\sigma(x_1,\dots))&\to&\delta(q_1(x_1), q_2(x_1), \dots, q_n(x_1))\\
q_0(e)&\to&e
\end{array}
\]
The transducer $M_2=(\{p\},\Sigma,\Sigma,p,\{p(e)\to e\})$ realizes the
translation $\tau_{M_2} = \{(e,e)\}$.
If the intersection of the $L(A_i)$ is empty, then there is no tree $s\in T_\Sigma$ such
that $M_{q_i}(s)$ is defined for all $i\in\{1,\dots, n\}$, i.e., 
the first rule displayed above is never applicable. 
Hence, in this case also $\tau_{M_1}=\{(e,e)\}$, i.e., the transducers $M_1,M_2$
are equivalent. If the intersection is non-empty, then there is an input tree $s$ such that
$M(s)=\delta(s,s,\dots, s)$. 
Thus, $M_1$ is equivalent to $M_2$ if and only if the intersection of the $L(A_i)$ is empty.
\end{proof}

\subsection{Streaming Tree Transducers}

The (deterministic) streaming tree transducers of Alur and d'Antoni are a new model
with the same expressive power as deterministic $\mso$ tree translations which in turn
realize the same translations as deterministic macro tree translations of linear 
size increase. The idea of the model is to use a finite set of variables which hold
partial outputs. These variables are updated during a single depth-first left-to-right 
traversal of the input tree. It is stated in Theorem~20 of~\cite{DBLP:journals/corr/abs-1104-2599}
that equivalence of streaming tree transducers can be decided in exponential time.
The idea of the proof is the same as the one in Theorem~\ref{theo:mttlsi}:
construct a context-free language $L^{a,b}$ and use its Parikhness to
check if $a^nb^n$ is in the language. For them, $L^{a,b}$ is represented by a 
pushdown automaton $A$, the number of states of which is exponential in the number
of variables of the given streaming tree transducer.
They mention that checking if $a^nb^n$ is in $L(A)$
can be done in $\nptime$ using~\cite{DBLP:journals/fuin/Esparza97,DBLP:conf/icalp/SeidlSMH04}.

\begin{theorem}[\cite{DBLP:journals/corr/abs-1104-2599}]
Equivalence of streaming tree transducers is decidable in $\conexptime$.
\end{theorem}

For the transducers that map strings to nested strings, that is, for streaming
string-to-tree transducers their construction yields a $\pspace$ bound
(Theorem~21 of~\cite{DBLP:journals/corr/abs-1104-2599}).

\begin{theorem}[\cite{DBLP:journals/corr/abs-1104-2599}]
Equivalence of streaming string-to-tree transducers is decidable in $\pspace$.
\end{theorem}

\subsection{Visibly Pushdown Transducers}

Visibly pushdown languages were defined by Alur and Madhusudan~\cite{DBLP:conf/stoc/AlurM04} as
a particular subclass of the context-free languages. In fact, they are
just regular tree languages in disguise.
Visibly pushdown transducers were introduced by 
Raskin and Servais~\cite{DBLP:conf/icalp/RaskinS08}.
They translate well-nested input strings into strings, during one left-to-right
traversal of the input. 
If the output strings are nested as well, then they describe tree
transformations. The expressive power of the resulting tree transformations
is investigated by Caralp, Filiot, Reynier, Servais, and 
Talbot~\cite{DBLP:journals/corr/CaralpFRST13}. 
Such transducers cannot copy nor swap the order of input trees. Thus, they
are $\mso$ definable. But they are incomparable to the top-down or bottom-up
tree translations, because they can translate a tree into its yield (string
of leaf labels from left to right).

\begin{theorem}[\cite{DBLP:conf/mfcs/FiliotRRST10}]
Equivalence of functional visibly pushdown transducers
is $\dexptime$-complete. For total such transducers the problem is
in {\sc Ptime}.
\end{theorem}

The $\dexptime$-completeness result extends to the case
of regular look-ahead, as shown in Section~8.4 of~\cite{DBLP:conf/sofsem/FiliotS12,Servais11}.
Staworko, Laurence, Lemay, and Niehren~\cite{DBLP:conf/fct/StaworkoLLN09} have considered the
equivalence problem for deterministic visibly pushdown transducers and show
that it can be reduced in $\Ptime$ to the homomorphic equivalence problem
on context-free grammars. The latter was shown by 
Plandowski~\cite{DBLP:journals/jcss/KarhumakiPR95,DBLP:conf/esa/Plandowski94}
to be solvable in $\Ptime$.
They show in~\cite{DBLP:conf/fct/StaworkoLLN09}
that for several related classes the problem is in $\Ptime$,
for instance, linear and order-preserving deterministic top-down and
bottom-up tree transducers. 

\begin{theorem}[\cite{DBLP:conf/fct/StaworkoLLN09}]
Equivalence of deterministic visibly pushdown transducers is
decidable in {\sc Ptime}.
\end{theorem}

\section{Conclusion}

We discussed the decidability of equivalence for three incomparable
subclasses of deterministic macro tree transducers:
top-down tree transducers, linear size increase $\mtt$s, and
monadic $\mtt$s.
For top-down tree transducers the proof either uses its bounded height-balance
property and constructs an automaton that keeps track of the balance.
Alternatively, such transducers may be transformed into their canonical normal form 
and then be checked for isomorphism. For these decision procedures it is not
``harmful'' that a top-down tree transducer can copy a lot and be
of exponential size increase, because the multiple copies of equivalent
transducers must be well-nested into each other (cf. Figure~\ref{fig:joost}).
This nesting property is not present for $\mtt$s, and in particular the 
bounded height-balance does not hold for $\mtt$s, even not for monadic ones.
Thus, other techniques are needed in these two cases.
For the linear size increase subclass of $\mtt$s we may use the Parikh
property of the corresponding output languages: the two transducers are
merged (``twinned'') to output $a^mb^n$ if, on the same input, one transducer
produces at position $m$ of its output the letter $a$ while the other
transducer produces at position $n$ the letter $b$. Since this output language
is Parikh, we may decide if it contains $a^nb^n$ which implies that the transducers
are not equivalent (because $a\not= b$).
For monadic $\mtt$s we use yet another technique:
we simulate the transducers by $\hdtol$ sequences. Since the sequence equivalence
problem for $\hdtol$ systems is decidable (not detailed here), the result follows.
It remains a deep open problem whether or not equivalence is decidable for
arbitrary deterministic macro tree transducers.
Even for $\mtt$s with monadic output, which are the same as deterministic
top-down tree-to-string transducers, it is open whether or not equivalence
is decidable. Note that the availability of a canonical normal form is a much
stronger result than the decidability of equivalence: for instance, equivalence
is easily decided for top-down transducers with look-ahead, but, for such
transducers we only know a canonical normal form in the 
total case for a fixed look-ahead automaton~\cite{DBLP:journals/corr/EngelfrietMS13}. 
In fact, even to decide
whether or not a given $\tr$ is equivalent to a $\ttt$ is a difficult
open problem; 
it was solved recently for a subclass of $\tr$s~\cite{DBLP:journals/corr/EngelfrietMS13}.

\bibliographystyle{eptcs}
\bibliography{bib} 
\end{document}

%% file: joost.pdf_t
\begin{picture}(0,0)%
\includegraphics{joost.pdf}%
\end{picture}%
\setlength{\unitlength}{3947sp}%
\begingroup\makeatletter\ifx\SetFigFont\undefined%
\gdef\SetFigFont#1#2#3#4#5{%
  \reset@font\fontsize{#1}{#2pt}%
  \fontfamily{#3}\fontseries{#4}\fontshape{#5}%
  \selectfont}%
\fi\endgroup%
\begin{picture}(5127,1824)(5236,-4873)
\put(5776,-4561){\makebox(0,0)[lb]{\smash{{\SetFigFont{10}{12.0}{\familydefault}{\mddefault}{\updefault}{\color[rgb]{0,0,0}$u$}%
}}}}
\put(5251,-4336){\makebox(0,0)[lb]{\smash{{\SetFigFont{10}{12.0}{\familydefault}{\mddefault}{\updefault}{\color[rgb]{0,0,0}$s$}%
}}}}
\put(8926,-3436){\makebox(0,0)[lb]{\smash{{\SetFigFont{10}{12.0}{\familydefault}{\mddefault}{\updefault}{\color[rgb]{0,0,0}$M_1(s)$}%
}}}}
\put(9076,-3661){\makebox(0,0)[lb]{\smash{{\SetFigFont{10}{12.0}{\familydefault}{\mddefault}{\updefault}{\color[rgb]{0,0,0}$=M_2(s)$}%
}}}}
\put(7726,-3811){\makebox(0,0)[lb]{\smash{{\SetFigFont{10}{12.0}{\familydefault}{\mddefault}{\updefault}{\color[rgb]{0,0,0}$p_1$}%
}}}}
\put(9001,-4561){\makebox(0,0)[lb]{\smash{{\SetFigFont{10}{12.0}{\familydefault}{\mddefault}{\updefault}{\color[rgb]{0,0,0}$p_2$}%
}}}}
\put(9001,-4111){\makebox(0,0)[lb]{\smash{{\SetFigFont{10}{12.0}{\familydefault}{\mddefault}{\updefault}{\color[rgb]{0,0,0}$q_3$}%
}}}}
\put(8926,-4411){\makebox(0,0)[b]{\smash{{\SetFigFont{10}{12.0}{\familydefault}{\mddefault}{\updefault}{\color[rgb]{0,0,0}$t'$}%
}}}}
\put(7651,-4261){\makebox(0,0)[b]{\smash{{\SetFigFont{10}{12.0}{\familydefault}{\mddefault}{\updefault}{\color[rgb]{0,0,0}$t$}%
}}}}
\put(7501,-4486){\makebox(0,0)[lb]{\smash{{\SetFigFont{10}{12.0}{\familydefault}{\mddefault}{\updefault}{\color[rgb]{0,0,0}$q_1$}%
}}}}
\put(7876,-4561){\makebox(0,0)[lb]{\smash{{\SetFigFont{10}{12.0}{\familydefault}{\mddefault}{\updefault}{\color[rgb]{0,0,0}$q_2$}%
}}}}
\end{picture}%